\documentclass[11pt]{article}
\usepackage[affil-it]{authblk}
\usepackage[usenames,dvipsnames]{xcolor}
\usepackage{amsfonts}
\usepackage{amsmath,amsthm,amssymb,amsfonts}
\usepackage{enumerate}
\usepackage[english]{babel}
\usepackage{graphicx}	
\usepackage[caption=false]{subfig}
\usepackage[margin=0.99in]{geometry}
\usepackage{url}
\usepackage{todonotes}
\usepackage{bbm}
\usepackage{bbold}
\usepackage{tikz}
\usetikzlibrary{chains}
\usetikzlibrary{fit}
\usepackage{pgflibraryarrows}		
\usepackage{pgflibrarysnakes}		
\usepackage{authblk}
\usepackage{makecell}

\usepackage{epsfig}
\usetikzlibrary{shapes.symbols,patterns} 
\usepackage{pgfplots}

\usepackage{hyperref}
\hypersetup{colorlinks=true,citecolor=blue,linkcolor=blue,filecolor=blue,urlcolor=blue,breaklinks=true}

\usepackage{nicefrac}
\usepackage{mathtools}

\usepackage{pgf, tikz}
\usetikzlibrary{arrows, automata}

\usepackage{enumitem}

\theoremstyle{plain}
\newtheorem{theorem}{Theorem}[section]
\newtheorem{lemma}[theorem]{Lemma}

\newtheorem{corollary}[theorem]{Corollary}
\newtheorem{proposition}[theorem]{Proposition}

\newtheorem*{theorem*}{Theorem}
\newtheorem*{proposition*}{Proposition}

\theoremstyle{definition}
\newtheorem{definition}[theorem]{Definition}
\newtheorem{remark}[theorem]{Remark}
\newtheorem{example}[theorem]{Example}

\newtheorem*{definition*}{Definition}

\newcommand*{\cI}{\mathcal{I}}

\newcommand*{\cX}{\mathcal{X}}
\newcommand*{\cY}{\mathcal{Y}}

\newcommand*{\M}{\mathbb{M}}

\newcommand*{\RR}{\mathbb{R}}

\newcommand*{\NN}{\mathbb{N}}

\newcommand*{\GL}{\mathrm{GL}}

\newcommand*{\rank}{\mathrm{rank}}

\newcommand*{\id}{\mathrm{I}}

\newcommand*{\supp}{\mathrm{supp}}

\newcommand{\ra}{\rightarrow}

\DeclarePairedDelimiterX{\inner}[2]{\langle}{\rangle}{#1, #2}
\DeclareMathOperator{\subnonrank}{Q_{+}}

\DeclareMathOperator{\leqk}{\leqslant_{+}}

\DeclareMathOperator{\asympsubnonrank}{\widetilde{Q}_{+}}

\DeclareMathOperator{\nrank}{R_{+}}
\DeclareMathOperator{\asynrank}{\widetilde{R}_{+}}

\DeclareMathOperator{\rankk}{R}
\DeclareMathOperator{\asymrank}{\widetilde{R}}
\DeclareMathOperator{\subrank}{Q}
\DeclareMathOperator{\asymsubrank}{\widetilde{Q}}

\title{On the Asymptotic Nonnegative Rank of Matrices and its Applications in Information Theory}

\author[1]{Yeow Meng Chee}
\author[2]{Quoc-Tung Le}
\author[1]{Hoang Ta}
\affil[1]{Department of Industrial Systems Engineering and Management, National University of Singapore, Singapore}
\affil[2]{Toulouse School of Economics, Université de Toulouse, France}

\begin{document}
	\maketitle
			\begin{abstract}
				In this paper, we study \emph{the asymptotic nonnegative rank} of matrices, which characterizes the asymptotic growth of the nonnegative rank of fixed nonnegative matrices under the Kronecker product. This quantity is important since it governs several notions in information theory such as the so-called \emph{exact Rényi common information} and the \emph{amortized communication complexity}. By using the theory of asymptotic spectra of V. Strassen (J. Reine Angew. Math. 1988), we define formally the asymptotic spectrum of nonnegative matrices and give a dual characterization of the asymptotic nonnegative rank. As a complementary of the nonnegative rank, we introduce the notion of the subrank of a nonnegative matrix and show that it is exactly equal to the size of the maximum induced matching of the bipartite graph defined on the support of the matrix (therefore, independent of the value of entries). Finally, we show that two matrix parameters, namely rank and fractional cover number, belong to the asymptotic spectrum of nonnegative matrices. 
			\end{abstract}
					
		\section{Introduction}
\label{sec:introduction}
The nonnegative rank of a nonnegative matrix $A$ of size $m \times n$, denoted by $\nrank(A)$, is the smallest integer number $r$ such there exist two nonnegative matrices $B$ and $C$ of size $m \times r$ and $r \times n$ respectively such that $A = BC$ ($BC$ is the usual matrix product between $B$ and $C$). The notion of nonnegative rank can be found in various fields such as combinatorial optimization~\cite{yannakakis1988expressing} and communication complexity~\cite{kushilevitz1997communication,lovasz1990communication,lee2009lower,rao2020communication}. For instances, in two parties communication complexity, two players compute a boolean function $f: \mathcal{X}\times \mathcal{Y} \to \{0,1\}$ on inputs $(x,y) \in \mathcal{X} \times \mathcal{Y}$ where each party holds only one argument (either $x$ or $y$). In this setting, the communication complexity $D(f)$ of the function $f$ is the minimum number of exchanging bits to evaluate $f$. Results in \cite{lovasz1990communication,lee2009lower} show that the logarithm of the nonnegative rank of the associated matrix $M_f$ of size $|\mathcal{X}| \times |\mathcal{Y}|$ defined by $M_f(x,y) = f(x,y)$ is a lower bound for the communication complexity, i.e., $\log \nrank(M_f) \leq D(f)$. Another application of nonnegative rank is in combinatorial optimization. In~\cite{yannakakis1988expressing}, it is shown that the minimum number of linear inequalities representing a given polytope can be determined by the nonnegative rank of a particular matrix. This quantity plays an important role in solving linear programming because the complexity of interior-point methods depends on the number of linear inequalities. Finally,  it is demonstrated in ~\cite{zhang2012quantum,jain2013efficient} that the minimum amount of information exchanged by two parties in order to sample from the bivariate probability distribution $p(x,y) = A_{x,y}$ represented by a matrix $A$ (assuming $A$ is normalized so that $\sum_{x,y}A_{x,y} = 1$) can be determined by the nonnegative rank of $A$.
\par

The asymptotic nonnegative rank, which is defined as the asymptotic growth of the nonnegative rank of a matrix under Kronecker product (see Equation~\eqref{eq:define_asymptotic_nrank} for its mathematical definition), is equally important since it is also related to many quantities in information theory and communication complexity such as the \emph{exact Rényi common information}~\cite{yu2020exact} and the \emph{amortized communication complexity}~\cite{feder1995amortized}. However, in comparison to the nonnegative rank, the asymptotic counterpart is not equally well-studied. 
For this reason, in this paper, we make several contributions to the study of asymptotic nonnegative rank: first, we give a dual characterization of this notion as well as additional results that relate to it via the theory of asymptotic spectra. Historically, this theory was developed by V. Strassen in~\cite{strassen1988asymptotic} (see also the exposition in {\cite[Part I]{wigderson2022asymptotic}}) to understand the arithmetic complexity of matrix multiplication (see, e.g.,~\cite{peterbook}). More recently, the theory of asymptotic spectra has been applied to a range of other problems in combinatorics \cite{zuiddam2019asymptotic_graph}, quantum information \cite{li2020quantum,farooq2023asymptotic}, circuit complexity \cite{robere2022amortized} and other domains \cite{jensen2019asymptotic,fritz2021generalization,Peter_vara_generlized}. Two related notions to the nonnegative rank and the asymptotic nonnegative rank, \emph{the nonnegative subrank} and its asymptotic version, are also introduced and studied. 


\subsection{Notations and Definitions}
\label{subsec:nation}
We define some basic notions that will be used in the rest of this paper. Let $A \in \RR_{+}^{m_A \times n_A},B \in \RR_{+}^{m_B \times n_B}$ be two nonnegative matrices, the \emph{direct sum} $A \oplus B$ and \emph{Kronecker product} $A \otimes B$ are defined respectively by: 
		\begin{align*}
			&A \oplus B \coloneqq \begin{pmatrix}
				A & 0 \\
				0 & B 
			\end{pmatrix} \in \RR_{+}^{(m_A + n_A) \times (m_B + n_B)}, \quad  \\ &A \otimes B \coloneqq \begin{pmatrix}
				A_{1,1}B & \ldots & A_{1,n_A}B\\
				\vdots & \ddots & \vdots\\
				A_{m_A,1}B & \ldots & A_{m_A,n_A}B
			\end{pmatrix} \in \RR_{+}^{m_Am_B \times n_An_B} \, ,
		\end{align*}
		where $A_{i,j}$ is the value of $A$ at the index $(i,j)$. 
		
		For $n \in \NN$, denote $\id_{n}$ the identity matrix of size $n$ and $\GL^+(n):= \{A \in \RR_{+}^{n \times n} \mid A^{-1} \in \RR_{+}^{n \times n}\}$ a subgroup of $\GL(n)$ whose elements and their inverses are nonnegative matrices.
		
		Given a set $S$ equipped with a binary operator $\cdot: S \times S$, we say $(S, \cdot)$ is a \emph{commutative monoid} if:
		\begin{enumerate}[label=(\arabic*)]
			\item The operator $\cdot$ is associative: $(a \cdot b) \cdot c = a \cdot (b \cdot c), \forall a, b, c \in S$.
			\item The operator $\cdot$ is commutative: $a \cdot b = b \cdot a, \forall a, b \in S$.
			\item There exists an identity element $0$: $a \cdot 0 = 0 \cdot a = a, \forall a \in S$.
		\end{enumerate}
		
		For a nonnegative matrix $A$, $\supp(A):= \{(i,j):A_{i,j}>0\}$ the support of $A$ is the set of indices whose coefficients in $A$ is strictly positive. In this work, $\supp(A)$ has two interpretations: either a subset of the index set or a binary matrix of the same size as $A$ whose coefficients at $(i,j) $ are equal to $1$ if $(i,j) \in \supp(A)$ and $0$ otherwise.
  \paragraph*{Asymptotic nonnegative rank}	
		The nonnegative rank is sub-multiplicative under Kronecker product. Indeed, given two nonnegative matrices $A,B$, suppose that $A = UV^{T}$ and $B = XY^{T}$, where $U,V,X,Y$ are nonnegative matrices, then $A\otimes B = (U\otimes X)(V\otimes Y)^{T}$. Therefore $\nrank(A\otimes B) \leq \nrank(A)\nrank(B)$. In general, $\nrank$ is not multiplicative since there exists $A \in \RR_+^{4 \times 4}$ such that $\nrank(A) = 4$ and $\nrank(A \otimes A) = 15 \neq \nrank(A)^2$ \cite{beasley2013communication}. For more similar examples, we refer to \cite{beasley2013communication,vandaele2016heuristics}). 
		Therefore, the \emph{asymptotic nonnegative rank} is interesting to consider. Mathematically, the asymptotic nonnegative rank of nonnegative matrices is defined as follows.
		\begin{equation}
                     \label{eq:define_asymptotic_nrank}
		    \asynrank(A) \coloneqq \lim_{n \ra \infty} \nrank(A^{\otimes n})^{1/n} \, .
		\end{equation}
		Since the sequence $\nrank(A^{\otimes })^{1/n}$ is sub-multiplicative, by Fekete's lemma~\cite{fekete1923verteilung}, we can prove that $\asynrank(A)$ is well-defined and the limit can be replaced by the infimum over $n$. Namely,
		\begin{align*}
			\asynrank(A) = \inf_{n} \nrank(A^{\otimes n})^{1/n} \, .
		\end{align*}
We also note that the approximated version of the asymptotic nonnegative rank was studied in~\cite{braun2017information}. In that work, the authors showed that this notion captures the \emph{common information}~\cite{wyner1975common}.
\subsection{Applications of asymptotic nonnegative rank}

The asymptotic nonnegative rank can be  interpreted as the \emph{exact Rényi common information of order zero}~\cite{yu2020exact}. This information is defined based on correlated sources (joint distribution), which is represented as a nonnegative matrix with normalized rows and columns. Specifically, given correlated sources $\pi_{XY}$ over $X \times Y$. The exact Rényi common information of order zero, denoted by $T_{Ex}^{0}(\pi_{XY})$, is the minimum common randomness rate for exact generation of the target distribution, with the constraint that common randomness can only be compressed using fixed-length codes. In~\cite{yu2020exact}, the authors showed that this information can be characterized by asymptotic nonnegative rank of $\pi_{XY}$. More precisely, they proved $T_{Ex}^{0}(\pi_{XY}) = \log \asynrank(\pi_{XY})$.  For more information about the exact Rényi common information, we refer to~\cite[Section 5 and Section 7]{CIT-122}.

\par 
 In addition, the asymptotic nonnegative rank can be used as a lower bound for the \emph{amortized communication complexity}. Unlike the original communication complexity (that we briefly introduce in the precedent paragraph), this model is interested in the minimum number of exchanging bits required to calculate a boolean function $f$ multiple times (on different inputs ${(x_i, y_i) \mid i = 1, \ldots, m}$). A trivial upper bound for this quantity is $m$ times the complexity of calculating $f$ once. In other words, we obtain this complexity by simply calculating $f(x_i, y_i)$ separately, one after another. However, it turns out that in many computational models, the amortized communication complexity is smaller than this natural upper bound. It implies that the gain of complexity by performing calculation in batch can be significant. A notable example can be found in \cite{feder1995amortized}, where the author constructs a specific function $f$ whose communication complexity $D(f)$ and its amortization $\bar{D}(f)$ are totally different, i.e., $\bar{D}(f)/m<D(f)$. Using the result $\log \nrank(M_f) \leq D(f)$ of \cite{lovasz1990communication}, we can prove a similar amortized counterpart, which states: $\log\asynrank(M_f) \leq \bar{D}(f)/m$ (remind that $M_f$ is of size $|\mathcal{X}| \times |\mathcal{Y}|, M_f(x,y) = f(x,y)$). Therefore, the asymptotic nonnegative rank $\asynrank$ is important for the study of amortized communication complexity.
\subsection{Contribution} In this paper, we consider the asymptotic spectra theory in the context of asymptotic nonnegative rank. Our main contributions are:
		\begin{enumerate}[leftmargin=*]
			\item By introducing a nonnegative matrix semiring $\M_+$ with two operators: $\oplus$ (the direct sum) and $\otimes$ (the Kronecker product) and a \emph{preorder} $\leqk$ (cf. Definition \ref{def:strassenorder}), we define $\mathbf{X}(\M_+,\leqk)$, the \emph{asymptotic spectrum of $\M_+$}. The set $\mathbf{X}(\M_+,\leqk)$ consists of \emph{monotone} semiring homomorphisms $\phi$ from the set of nonnegative matrices $\M_+$ to the set of nonnegative numbers $\RR_{\geq 0}$, i.e., if $A \leqk B$, then $\phi(A) \leq \phi(B)$ ($A, B$ are two nonnegative matrices). In fact, $\mathbf{X}(\M_+,\leqk)$ is the nonnegative matrix version of other existing asymptotic spectra such as the asymptotic spectrum of tensors~\cite{strassen1988asymptotic}, graphs \cite{zuiddam2019asymptotic_graph} and LOCC transformations~\cite{jensen2019asymptotic}, which are studied in the literature. 
			\item The definition of $\mathbf{X}(\M_+,\leqk)$ allows us to establish another characterization of the asymptotic nonnegative rank of a matrix $A$ in Theorem \ref{thm:duality}, which states: 
			\begin{equation*}
				\asynrank(A) = \max_{\phi \in \mathbf{X}(S,\leqk)} \phi(A)
			\end{equation*}
			\item Similar to previous works on asymptotic spectra, we define and study two quantities: the \emph{subrank} ($\subnonrank$) and \emph{asymptotic subrank} ($\asympsubnonrank$) of a nonnegative matrix. While $\subnonrank$ is a natural lower bound of $\nrank$, $\asympsubnonrank$ is the dual of $\asynrank$. In this work, we prove that $\subrank(A)$  is exactly equal to the induced matching number of the bipartite graph defined on the support of $A$ (see Theorem \ref{thm:nonnegative_subrank}) and thus, independent of the coefficients of $A$. Moreover, similar to the asymptotic nonnegative rank, the asymptotic subrank can also be characterized by using $\mathbf{X}(\M_+,\leqk)$ (cf. Theorem \ref{thm:duality}) as:
			\begin{equation*}
				\asympsubnonrank(A) = \min_{\phi \in \mathbf{X}(S,\leqk)} \phi(A)
			\end{equation*}
		\item Since $\mathbf{X}(\M_+,\leqk)$ allows us to characterize precisely $\asynrank$ and $\asympsubnonrank$, we study in detail the set $\mathbf{X}(\M_+,\leqk)$. Our investigation of $\mathbf{X}(\M_+,\leqk)$ consists of showing its two elements: the rank function and the fractional cover number (cf. Section \ref{sec4:elements_spectral_point}). We note that the construction of explicit elements in the asymptotic spectrum of a specific instance is nontrivial. For example, the construction of nontrivial (universal) elements in the asymptotic spectrum of tensors has been an open problem for more than thirty years until it was solved by recent work in~\cite{christandl2018universal}. In addition, we also showed that for all $\phi \in \mathbf{X}(\M_+,\leqk)$, for all \emph{non-degenerate} lower triangular nonnegative matrix $A \in \RR^{m \times n}$, we have: $\phi(A) = \min(m,n)$.
		\end{enumerate}
 \paragraph{Organization} The paper is organized as follows: In Section~\ref{sec2:spectra}, we briefly present the theory of asymptotic spectra of Strassen. We then introduce the asymptotic spectrum on nonnegative matrices and give a dual characterization of the asymptotic nonnegative rank of nonnegative matrices in Section~\ref{sec3:spectrum_nonnegative_matrices}. Finally, in Section~\ref{sec4:elements_spectral_point}, we show two elements in the asymptotic spectrum of nonnegative matrices.

		
		\section{Asymptotic spectra and Strassen's spectral theorem} 
		\label{sec2:spectra}
		In this section, we present briefly the abstract theory of asymptotic spectra which was introduced by V. Strassen~\cite{strassen1988asymptotic}. For a more complete presentation, we refer the readers to~\cite{strassen1988asymptotic,zuiddam2018algebraic,wigderson2022asymptotic}. 
		
		Let $(S,+,\cdot,0_S,1_S)$ be a commutative semiring, i.e., a set $S$ equipped with a binary addition operation $+$, a binary multiplication $\cdot$, and having two elements $0_S,1_S \in S$  such that for all $a,b,c \in S$:
		\begin{enumerate}[label=(\arabic*)]
			\item $(S,+)$ is a commutative monoid with the identity element $0_S$.
			\item $(S,\cdot)$ is a commutative monoid with the identity element $1_S$.
			\item $\cdot$ distributes over $+$: $a \cdot (b+c) = (a\cdot b) + (a\cdot c)$
			\item $0_S \cdot a = 0_S$. 
		\end{enumerate}
		For $n \in \NN$ we denote the sum of $n$ elements $1_S$, i.e., $1_S +\dots +1_S \in S$ by $n_S$. A \emph{preorder} $\leqslant$ on $S$ is a relation on $S$ such that for all $a,b,c \in S$ 
		
		\begin{enumerate}[label=\arabic*)]
			\item $\leqslant$ is reflexive: $a\leqslant a$
			\item $\leqslant$ is transitive: $a\leqslant b$ and $b \leqslant c$ implies $a\leqslant c$.
		\end{enumerate}
		
		\begin{definition}
			\label{def:strassenorder}
			A preorder $\leqslant$ on $S$ is a \emph{Strassen preorder} if
			\begin{enumerate}[label = (\roman*)]
				\item For all $n,m \in \NN$ $n \leq m$ iff $n_S \leqslant m_S$ in $S$.
				\item For all $a,b,c,d \in S$ if $a \leqslant b$ and $c\leqslant d$, then $a+c \leqslant b+d$ and $a \cdot c \leqslant b \cdot d$.
				\item For all $a,b \in S, b\neq 0$ there is $r \in \NN$ such that $a \leqslant r_S \cdot b$. 
			\end{enumerate}
		\end{definition}   
		Let $S$ be a commutative semiring and let $\leqslant$ be a Strassen preorder on $S$. Throughout the paper, we will use $\leq$ to denote the usual order on $\RR$. Let $\RR_{\geq 0}$ be the semiring of nonnegative real numbers. A semiring homomorphism from $S$ to $\RR_{\geq 0}$ (denoted by $\operatorname{Hom}\left(S, \mathbb{R}_{\geq 0}\right)$) is a map $\phi: S \to \RR_{\geq 0}$ such that: 
		\begin{enumerate}[label = (\arabic*)]
			\item $\phi(a \cdot b) = \phi(a)\phi(b)$
			\item $\phi(a + b) = \phi(a)+\phi(b)$
			\item $\phi(1_S) = 1$.
		\end{enumerate}
		Let $\mathbf{X}(S,\leqslant)$ be the set of $\leqslant$-monotone semiring homomorphisms from $S$ to $\RR_{\geq 0}$, i.e:
		\begin{align*}
			\mathbf{X}(S, \leqslant) \coloneqq \left\{\phi \in \operatorname{Hom}\left(S, \mathbb{R}_{\geq 0}\right): \forall a, b \in S, a \leqslant b \Rightarrow \phi(a) \leq \phi(b)\right\} \,,
		\end{align*}
		we note that for every $\phi \in \mathbf{X}(S,\leqslant)$ holds $\phi(1_S) = 1$ and thus $\phi(n_S) = n$ for all $n \in \NN$. We call $\mathbf{X}(S,\leqslant)$ the \emph{asymptotic spectrum} of $(S,\leqslant)$ and elements of $\mathbf{X}(S,\leqslant)$ \emph{spectral points}.

		Let $a \in S$, the rank of $a$ is defined as $\mathrm{R}(a) \coloneqq \min\{n \in \NN: a \leqslant n_S \}$. The subrank of $a$ is defined as $\mathrm{Q}(a) \coloneqq \max \{n \in \NN: n_S \leqslant a\}$. The asymptotic subrank and asymptotic rank of $a$ are defined as
		\begin{equation}
			\label{eq:asymptotic_rank}	
			\asymrank(a) \coloneqq \lim_{N \to \infty} \rankk(a^N)^{1/N} \text{ and } \asymsubrank(a) \coloneqq \lim_{N \to \infty}\subrank(a^N)^{1/N} 
		\end{equation}
		From Fekete's lemma, the limits in \eqref{eq:asymptotic_rank} indeed exist and can be replaced by an infimum and a supremum, that is, 
		
		\begin{align*}
			\asymrank(a) = \inf_{N} \rankk(a^N)^{1/N} \text{, and } \asymsubrank(a) = \sup_{N} \subrank(a^N)^{1/N} \, .
		\end{align*}

			The notions of asymptotic rank and asymptotic subrank have been extensively studied in various fields. For instance, in~\cite{strassen1988asymptotic}, the author showed that the asymptotic rank of matrix multiplication tensors corresponds to the complexity of matrix multiplication, more examples can be found in~\cite{wigderson2022asymptotic}. Beside that, the notion of asymptotic subrank for tensors can be utilized to establish upper bounds for several combinatorial problems, such as capset and sunflower problems~\cite{christandl2018universal}, as well as the Shannon capacity of hypergraphs~\cite{christandl_et_al:LIPIcs.ITCS.2022.48}. Furthermore, research has also focused on exploring structures of the set of nonnegative real numbers obtained from the asymptotic rank or asymptotic subrank of complex tensors, as seen in~\cite{blatter2022countably, christandl2022gap,gesmundo2023next}.	
		
		In~\cite{strassen1988asymptotic}, the author showed that the asymptotic rank and asymptotic subrank have the following dual characterization in terms of the asymptotic spectrum.
		
		\begin{theorem}[{\cite[Th. 3.8]{strassen1988asymptotic}} and {\cite[Cor. 2.14]{zuiddam2018algebraic}}]
			\label{thm:Strassen_asymptotic}
			Let $S$ be a commutative semiring and let $\leqslant$ be a Strassen preorder on $S$. For any $a\in S$ such that $1_S \leqslant a$ and $2_S \leqslant a^k$ for some $k \in \NN$, we have that
			\begin{align*}
				\asymrank(a)  = \max_{\phi \in \mathbf{X}(S,\leqslant)} \phi(a) \text{, and } \asymsubrank(a) = \min_{\phi \in \mathbf{X}(S,\leqslant)}\phi(a) \, .
			\end{align*}
		\end{theorem}

	The theory of asymptotic spectra, as well as Theorem~\ref{thm:Strassen_asymptotic}, was first applied to study the complexity of the matrix multiplication problem~\cite{strassen1988asymptotic,Str91}. Subsequently, in~\cite{zuiddam2019asymptotic_graph,li2020quantum}, the authors used this theory to introduce the (quantum) asymptotic spectrum of graphs, allowing them to establish a dual characterization for the (quantum) Shannon capacity on graphs. In the context of quantum information theory, \cite{jensen2019asymptotic} used this theory to provide a characterization of the converse error exponents for asymptotic entanglement transformations. In the following section, we will build a commutative semiring for the set of nonnegative matrices in which the theory of asymptotic spectra can be applied. This will allow us to establish a characterization for the asymptotic nonnegative rank.

		\section{Asymptotic spectrum of nonnegative matrices and duality}
		\label{sec3:spectrum_nonnegative_matrices}
		\subsection{Asymptotic spectrum of nonnegative matrices}
		\label{subsec: asymptotic_spectrum_matrices}
		In this section, our goal is to apply the theory of the previous section to give the dual characterization of the asymptotic nonnegative rank of matrices. We remind the readers that $\otimes$ denotes the Kronecker product, $\oplus$ denotes the direct sum of matrices, and $\id_{n}$ denotes the identity matrix of size $n \times n$.     
		
		A natural semiring can be defined by considering the set of nonnegative matrices equipped with the binary operators $\oplus$ and $\otimes$. However, the Kronecker product $\otimes$ is not commutative. Therefore, in the following construction, rather than working with the set of nonnegative matrices, we will work with its \emph{quotient} space where $\oplus$ and $\otimes$ are commutative. In fact, we will introduce an equivalence relation in which equivalent elements share the same nonnegative rank, and the induced operators $\oplus/\otimes$ on the corresponding quotient space are commutative. This allows us to bypass the technical requirement and apply the theory of asymptotic spectra.
		\par
		

		Denote by $\M(m,n)$ the set of all nonnegative matrices of size $m \times n$, for two nonnegative matrices $A,B \in \M(m,n)$, we write $A\cong B$ if there is a pair of matrices $(X,Y) \in \GL^+(m) \times \GL^+(n)$ such that $A = XBY^\top$ (or equivalently $B = X^{-1}A(Y^{-1})^\top$). It is clear that $\cong$ is an equivalence relation since if $A = X_1BY_1^T, B = X_2CY_2^\top$, then $A = (X_1X_2)C(Y_1Y_2)^\top$ with $X_1X_2 \in \GL^+(m), Y_1Y_2 \in \GL^+(n)$. In addition, the relation $\cong$ enjoys other nice properties, as shown in Lemma \ref{lem:congproperty}.
		\begin{lemma}
			\label{lem:congproperty}
			Given four matrices $A, B \in \M(m_1,n_1)$, $C, D \in \M(m_2,n_2)$ such that $A \cong B, C \cong D$, we have:
			\begin{enumerate}[leftmargin = *]
				\item $\nrank(A) = \nrank(B)$ (and $\nrank(C) = \nrank(D)$).
				\item $A \oplus C \cong B \oplus D$.
				\item $A \otimes C \cong B \otimes D$.
			\end{enumerate}
		\end{lemma}
		\begin{proof}
			Since $A \cong B$, so there are $X \in \GL^{+}(m_1)$ and $Y \in \GL^{+}(n_1)$ such that $XAY^{T} = B$. This implies that $\nrank(A) = \nrank(B)$ since $X,Y$ are invertible matrices. 
			
			Since $C \cong D$, so there are $U \in \GL^{+}(m_2)$ and $V \in \GL^{+}(n_2)$ such that $UCV^{T} = D$. We have $(X \oplus U)(A\oplus C)(Y\oplus V)^{T} = B \oplus D$, therefore $A \oplus C \cong B \oplus D$ since $X \oplus U \in \GL^{+}(m_1 +m_2)$ and $Y \oplus V \in \GL^{+}(n_1 +n_2)$. Similarly, one has $(X \otimes U)(A \otimes C)(Y \otimes V)^{T} = B \otimes D$, this implies $A \otimes C \cong B \otimes D$. 
			
		\end{proof}
		
		Let $A \in M(m_A,n_A)$ and $B \in M(m_B,n_B)$, we say $A$ and $B$ are \emph{equivalent}, and write $A \sim B$, if there are zero matrices $A_0 = 0 \in \RR_{+}^{a_1 \times a_2}$ and $B_0 = 0 \in \RR_{+}^{b_1\times b_2}$ such that $A\oplus A_0 \cong B \oplus B_0$. In fact, $\sim$ is also an equivalence relation and elements in the same equivalence class shares a common nonnegative rank, as shown below.
		\begin{lemma}
			\label{lem:equivrelation}
			The binary relation $\sim$ is an equivalence one. Moreover, the nonnegative rank is invariant under $\sim$.
		\end{lemma} 
		\begin{proof}
			We first prove that $\sim$ is an equivalence class. Since the reflexivity and symmetry of $\sim$ is trivial, we focus on proving the transitivity of $\sim$, i.e, for three matrices $A, B, C$, if $A \sim B, B \sim C$, then $A \sim C$. Indeed, by definition of $\sim$, there exists zero matrices $A_0, B_0, B_1, C_0$ such that:
			\begin{equation*}
				A \oplus A_0 \cong B \oplus B_0, \qquad B \oplus B_1 \cong C \oplus C_0
			\end{equation*}
			Therefore, 
			\begin{equation*}
				(A \oplus A_0) \oplus B_1 \cong (B \oplus B_0) \oplus B_1, \qquad (B \oplus B_1) \oplus B_0 \cong (C \oplus C_0) \oplus B_0,
			\end{equation*}
			using Lemma \ref{lem:congproperty}. Moreover, $B_0 \oplus B_1 = B_1 \oplus B_0$ since both are zero matrix of the same sizes. This implies $B \oplus B_1 \oplus B_0 = B \oplus B_0 \oplus B_1$. Using the fact that $\cong$ is an equivalence relation, we have:
			\begin{equation*}
				A \oplus (A_0 \oplus B_1) \cong C \oplus (C_0 \oplus B_0)
			\end{equation*}
			Notice that both $A_0 \oplus B_1$ and $C_0 \oplus B_0$ are zero matrices. This proves the first claim.
			
			Moreover, the direct sum of a matrix with a zero matrix does not affect the nonnegative rank. Therefore, if $A \sim B$, then $\nrank(A) = \nrank(A \oplus A_0) = \nrank(B \oplus B_0) = \nrank(B)$. This concludes the proof. 
		\end{proof}
		
		Let $\M_{+}$ be the set of $\sim$-equivalence classes of nonnegative matrices (which is now well-defined thanks to Lemma \ref{lem:equivrelation}) , we define the addition and multiplication on $\M_{+}$ as follow
		\begin{align*}
			[A] \oplus [B] &= [A\oplus B] \\
			[A] \otimes [B] &= [A \otimes B] \, .
		\end{align*} 
		where $[A]$ denotes the equivalence class of the matrix $A$. These two operations are also well-defined, as shown in the following lemma.
		\begin{lemma}
			\label{lem:welldefinedoperator}
			Consider two pairs of matrices $(A,B)$ and $(A', B')$ such that $A \sim A'$ and $B \sim B'$, then:
			\begin{align*}
				[A \oplus B] &= [A' \oplus B']\\
				[A \otimes B] &= [A' \otimes B']\qedhere
			\end{align*}
		\end{lemma}
		\begin{proof}
			Due to the definition of the equivalence relation $\sim$, there exists zero matrices $A_0,B_0,A'_0,B'_0$ such that:
			\begin{align*}
				A \oplus A_0 \cong A' \oplus A'_{0}, \qquad B \oplus B_0 \cong B' \oplus B'_0 \, .
			\end{align*}
			By Lemma \ref{lem:congproperty}, we have:
			\begin{align}
				(A \oplus A_0) \oplus (B \oplus B_0) \cong (A' \oplus A'_0) \oplus (B' \oplus B'_0) \label{eq:pluscongwelldefine}\\
				(A \oplus A_0) \otimes (B \oplus B_0) \cong (A' \oplus A'_0) \otimes (B' \oplus B'_0) \label{eq:timecongwelldefine}
			\end{align}
			\begin{enumerate}[leftmargin=*]
				\item For addition $(\oplus)$: The key observation for this part is $(A \oplus A_0) \oplus (B \oplus B_0) \cong (A \oplus B) \oplus (A_0 \oplus B_0)$. Indeed, we have:
				\begin{equation*}
					(A \oplus A_0) \oplus (B_0 \oplus B_0) = \begin{pmatrix}
						A & 0 & 0 & 0\\
						0 & A_0 & 0 & 0\\
						0 & 0 & B & 0\\
						0 & 0 & 0 & B_0
					\end{pmatrix}, (A \oplus A_0) \oplus (B_0 \oplus B_0) = \begin{pmatrix}
						A & 0 & 0 & 0\\
						0 & B & 0 & 0\\
						0 & 0 & A_0 & 0\\
						0 & 0 & 0 & B_0
					\end{pmatrix}
				\end{equation*}
				Therefore, the matrix $(A \oplus B) \oplus (A_0 \oplus B_0)$ can be obtained by permuting rows and columns of $	(A \oplus A_0) \oplus (B_0 \oplus B_0)$. This implies there exists two permutation matrices $X, Y$ such that $X((A \oplus A_0) \oplus (B_0 \oplus B_0))Y^\top = (A \oplus B) \oplus (A_0 \oplus B_0)$. Similarly, we also have: $(A' \oplus A'_0) \oplus (B' \oplus B'_0) \cong (A' \oplus B') \oplus (A'_0 \oplus B'_0)$. Thus, we have $(A \oplus B) \oplus (A_0 \oplus B_0) \cong (A' \oplus B') \oplus (A'_0 \oplus B'_0)$ (by combining with \eqref{eq:pluscongwelldefine}). The proof is concluded by noticing that $A_0 \oplus B_0$ and $A_0' \oplus B'_0$ are zero matrices.
				
				\item For multiplication $(\otimes)$: we have: 
				\begin{align*}
					(A \oplus A_0) \otimes (B \oplus B_0) &= [A \otimes (B \oplus B_0)] \oplus [A_0 \otimes (B \oplus B_0)]\\
					&\cong [(B \oplus B_0) \otimes A] \oplus [(B \oplus B_0) \otimes A_0]   
					\\ &\cong (A \otimes B) \oplus (A \otimes B_0 \oplus A_0 \otimes B  \oplus A_0 \otimes B_0)
				\end{align*}
				Here we use the fact $(A\oplus B)\otimes C = (A \otimes C) \oplus (B \otimes C)$  and $A \otimes B \cong B \otimes A$ (note that $A \otimes (B \oplus C) \neq (A \otimes B) \oplus (A \otimes C)$ in general). Similarly, we have  
				\begin{align*}
					(A' \oplus A'_0) \otimes (B' \oplus B'_0) \cong (A' \otimes B') \oplus (A' \otimes B'_0 \oplus A'_0 \otimes B'  \oplus A'_0 \otimes B'_0) \, .
				\end{align*}

				Note that the second terms of both above equations are zero matrices. Combining with \eqref{eq:timecongwelldefine}, we deduce that $[A \otimes B] = [A' \otimes B']$. \qedhere
			\end{enumerate}
		\end{proof}

		From above arguments, we conclude that the elements $[A] \oplus [B]$ and $[A] \otimes [B]$ does not depend on the choice of representative $A,B$. Hence, these two operations are well-defined in the quotient space $\M_{+}$. Moreover, one can easy to verify that $[A \oplus B] = [B \oplus A], [A \otimes B] = [B \otimes A]$. Thus, $(\M_{+},\oplus,\otimes,[0], [\id_{1}])$ is a commutative semiring where $[0]$ is the equivalence class of all zero matrices. 
		\begin{lemma}
			The set $\M_{+}$ with addition $\oplus$, multiplication $\otimes$, additive unit $[0]$ and multiplicative unit $[\id_{1}]$ is commutative semiring.
		\end{lemma}
		
		We write $A \leqk B$ if there are two nonnegative matrices $X \in \RR_{+}^{m_A \times m_B}$ and $Y \in \RR_{+}^{n_A \times n_B}$ such that $A = XBY^{T}$. This preorder ensures that if $A \leqk B$, then $\nrank(A) \leq \nrank(B)$. Note that $\leqk$ is less restrictive than $\cong$ (which requires $X \in \GL^+(m_A), Y \in \GL^+(n_A)$). This also implies that if $A \cong B$, then $A \leqk B$ and $B \leqk A$. The natural induced preorder of $\leqk$ on $\M_{+}$ is also well-defined $\M_{+}$, which is shown as follow:
		\begin{lemma}
			\label{lem:welldefineorder}
			If $A \leqk B$, then $A' \leqk B'$, for all $A' \sim A, B' \sim B$.
		\end{lemma}
		\begin{proof}
			It is sufficient to show that if $A \sim A'$, then $A \leqk A', A' \leqk A$. Indeed, combining with the transitivity of $\leqk$, we can conclude:
			\begin{equation*}
				A' \leqk A \leqk B \leqk B'
			\end{equation*}
			By definition of $\sim$, there exists two zero matrices $A_0, A'_0$ such that $A \oplus A_0 \cong A' \oplus A'_0$. Assume that $A \in \RR_{+}^{m\times n}$ and $A_0 = 0 \in \RR_{+}^{p\times q}$. Define matrices $X = \left[\begin{smallmatrix} \id_{m} & 0 \end{smallmatrix}\right] \in \RR_{+}^{m \times (m+p)}$ and $Y = \left[\begin{smallmatrix} \id_{n} & 0 \end{smallmatrix}\right] \in \RR_{+}^{n \times (n+q)}$. We have, $A = X(A \oplus A_0)Y^{T}$ and $A \oplus A_0 = X^{T}AY$. Therefore,
			\begin{align*}
				A \leqk A &\oplus A_0 \\
				A \oplus A_0 &\leqk A \,. 
			\end{align*}   
			Similarly, we also have: $	A' \leqk A' \oplus A'_0, A' \oplus A'_0 \leqk A'$. As such, $A \leqk A \oplus A_0 \cong A' \oplus A'_0 \leqk A'$. The other direction can be proved similarly. This yields the desired result.
		\end{proof}
		Hereinafter, to ease the notation, we abuse the notations by writing $A$ to denote both the matrix $A$ and also the equivalence class $[A]$. The meaning should be clear from the context. 
		
		The following lemma proves that $\leqk$ is also a Strassen preorder.
		
		\begin{lemma}
			The preorder $\leqk$ is a Strassen preorder on $\M_{+}$. That is, for nonnegative matrices $A,B,C,D \in \M_{+}$ we have the following.
			\begin{enumerate}[label=(\roman*)]
				\item { \label{eqq: state1}
					For all $n,m \in \NN$, $n \leq m$ if and only if $\id_{n} \leqk \id_{m}$
				}
				\item { \label{eqq: state2} 
					If $A \leqk B$ and $C \leqk D$ then $A \oplus C \leqk B \oplus D$ and $A\otimes C \leqk B \otimes D$
				}
				\item { \label{eqq: state3}
					If $B \neq 0$, there is an $r \in \NN$ such that $A \leqk \id_{r} \otimes B$.}  
				
			\end{enumerate} 
		\end{lemma}	
		\begin{proof}
			The statement \eqref{eqq: state1} is trivial.
			
			We prove \eqref{eqq: state2}. Let $U,V,X,Y$ such that $A = UBV^{T}$ and $C = XDY^{T}$. Then $A\oplus C = (U\oplus X)(A\oplus C)(V \oplus Y)^{T}$ and $A\otimes C = (U\otimes X)(B\otimes D)(V\otimes Y)^{T}$. This proves \eqref{eqq: state2}. 
			
			We prove \eqref{eqq: state3}. Let $r = \nrank(A)$. Then $A \leqk \id_{r}$. By assumption, $B \neq 0$, so $I_1 \leqk B$. Therefore $A \leqk \id_{r} \cong \id_{r} \otimes \id_{1} \leqk \id_{r} \otimes B$ (where $r = \nrank(A)$). This proves \eqref{eqq: state3}.
		\end{proof}
		We thus have a semiring $\M_{+}$ with a Strassen preorder $\leqk$. Thus, we can apply the theory of asymptotic spectra in this context. Let us translate the abstract terminology to this setting. 
		

		Let $S\subseteq \M_{+}$ be a semiring. For example, one may take $S = \M_{+}$, or one may choose any set $X \subseteq \M_{+}$ and let $S = \NN[X]$ be the subsemiring of $\M_{+}$ generated by $X$ under $\oplus$ and $\otimes$. 
		
		The asymptotic spectrum of $S$ is the set $\mathbf{X}(S,\leqk)$ of $\leqk$-monotone semiring homomorphism $S\to \RR_{\geq 0}$, i.e., all maps $\phi: S\to \RR_{\geq 0}$ such that. for all $A,B \in S$
		\begin{enumerate}[label = (\arabic*)]
			\item $\phi(A \otimes B) = \phi(A)\phi(B)$
			\item $\phi(A \oplus B) = \phi(A)+\phi(B)$
			\item $\phi(\id_{n}) = n$ for all $n \in \NN$
			\item If $A \leqk B$ then $\phi(A) \leq \phi(B)$.
		\end{enumerate}
		We call $\mathbf{X}(\M_{+},\leqk)$ the asymptotic spectrum of nonnegative matrices. In the next subsection, we use the asymptotic spectrum of nonnegative matrices to give a dual characterization for asymptotic nonnegative rank.  
		
		\subsection{Asymptotic nonnegative rank and asymptotic nonnegative subrank}
		From Subsection~\ref{subsec: asymptotic_spectrum_matrices}, we have shown that the preorder $\leqk$ is a Strassen preorder. Thus, similar to Section~\ref{sec2:spectra}, using the preorder $\leqk$, we can rewrite the nonnegative rank and define the nonnegative subrank of a nonnegative matrix $A$ as
		\begin{align*}
			\nrank(A) &\coloneqq \min \{r \in \NN: A \leqk \id_{r}\} \\
			\subnonrank (A) &\coloneqq \max \{ q \in \NN: \id_{q} \leqk A \} \, . 
		\end{align*} 
		
		Similarly to the asymptotic nonnegative rank,  the \emph{asymptotic nonnegative subrank} is naturally defined as follows.
		\begin{align*}
			\asympsubnonrank(A) &\coloneqq \lim_{n \ra \infty} \subnonrank(A^{\otimes n})^{1/n} = \sup_{n}\subnonrank(A^{\otimes n})^{1/n} \, .
		\end{align*}
		
		

		Using the abstract theory of Strassen, the asymptotic spectrum of nonnegative matrices can be used to characterize the asymptotic nonnegative rank and the asymptotic nonnegative subrank. 
		
		\begin{theorem}
			\label{thm:duality}
			Let $A$ be a nonnegative matrix. Then
			\begin{align*}
				\asynrank(A) = \max_{\phi \in \mathbf{X}(\M_{+},\leqk)} \phi(A) \text{, and }
				\asympsubnonrank(A) = \min_{\phi \in \mathbf{X}(\M_{+},\leqk)} \phi(A) \, .
			\end{align*} 
		\end{theorem}
		
		Before proceeding to the proof of Theorem~\ref{thm:duality}, we need to show first a combinatorial interpretation of subrank $\subnonrank$. It is easy to see that the nonnegative rank of matrix $A$ depends on its entries. In the other direction, we show in the following that the nonnegative subrank of a nonnegative matrix depends only on its support. We first introduce some notions and preliminary results. 
		
		Let $A$ be a nonnegative matrix, we remind that $\supp(A) = \{(i,j):A_{i,j}>0\}$ the support of $A$. Two pairs $(a,b),(c,d) \in \supp(A)$ is called \emph{totally pairwise independent} if $A_{a,d} = A_{c,b} = 0$. A subset $\Phi$ of $\supp(A)$ is called \emph{induced matching} of $A$ if each pair $(a,b) \neq (c,d) \in \Phi$ is totally pairwise independent. The name induced matching here come from the fact that the totally pairwise independent set of matrix $A$ is an induced matching of the bipartite graph that construct from the support of $A$. We denote $\gamma(A)$ to be the largest induced matching of $A$. Let us illustrate this notion by the following example. 
		\begin{example}
			\begin{equation*}
				A = 
				\begin{pmatrix}
					1 & 1 & 0 & 0 \\
					1& 0 & 1 & 0 \\
					0 & 1 & 0 & 1 \\
					1 & 0 & 1 & 1
				\end{pmatrix} \, .
			\end{equation*}
			It is to easy to verify that in this example, $\psi = \{(1,1),(2,3),(3,2),(4,4)\}$ is not an induced matching of $A$ (since $A_{2,1} = 1 \neq 0$), and $\Phi = \{(4,1),(3,2)\}$ is an induced matching of $A$. 
		\end{example}
		
		We show that for every nonnegative matrix, the largest size of an induced matching is equal to the nonnegative subrank, i.e.,
		\begin{theorem}
			\label{thm:nonnegative_subrank}
			For all nonnegative matrices $A$, $\subnonrank(A) = \gamma(A)$.
		\end{theorem}
		To prove Theorem~\ref{thm:nonnegative_subrank}, we first prove two results: Proposition~\ref{prop:subrankbound} and Proposition~\ref{prop:ordermatching}
		
		
		\begin{proposition}
			\label{prop:subrankbound}
			For any nonnegative matrix $A$, we have $\gamma(A) \leq \subnonrank(A)$.  
		\end{proposition}
		\begin{proof}
			We prove: $\gamma(A) \leq \subnonrank(A), \forall A \in \RR_{+}^{m\times n}$. Let $\{(x_1,y_1),\dots,(x_t,y_t) \} \subseteq \supp(A)$ be the maximum induced matching set of $A$. Due to the invariance of $\gamma$ and $\subnonrank$ under the permutation of rows and columns, we can assume WLOG that $x_i = y_i = i, \forall 1 \leq i \leq t$ (otherwise, we permute the rows and columns correspondingly). This allows us to write the matrix $A$ in the form of block matrices:
			\begin{equation*}
				A = \begin{pmatrix}
					D & A_1\\
					A_2 & A_3
				\end{pmatrix},
			\end{equation*} 
			where $D \in \RR^{t \times t}_+$ is a positive diagonal matrix and $A_i, i = 1,2,3$ are some nonnegative matrices. This claim holds since $\{(i, i) \mid i  = 1, \ldots, t\} \subseteq \supp(A)$ is an induced matching. Consider:
			\begin{equation}
				\label{eq:It}
				\begin{pmatrix}
					D^{-1/2} & 0
				\end{pmatrix}A \begin{pmatrix}
					D^{-1/2} \\ 0
				\end{pmatrix} = \begin{pmatrix}
					D^{-1/2} & 0
				\end{pmatrix}\begin{pmatrix}
					D & A_1\\
					A_2 & A_3
				\end{pmatrix} \begin{pmatrix}
					D^{-1/2} \\ 0
				\end{pmatrix} = I_t,
			\end{equation}
			where $D^{-1/2} \in \RR_+^{t \times t}$ is a positive diagonal matrix obtained by taking the root square of diagonal elements of $D$. Since $\left[\begin{smallmatrix} D^{-1/2} & 0 \end{smallmatrix}\right], \left[\begin{smallmatrix} D^{-1/2} \\ 0 \end{smallmatrix}\right]$ are nonnegative, Equation \eqref{eq:It} shows that $t = \gamma(A) \leq \subnonrank(A)$.
		\end{proof}
		
		\begin{proposition}
			\label{prop:ordermatching}
			Let $A,B,C$ be three nonnegative matrices such that $A = BC$, we have:
			\begin{align*}
				\gamma(A) \leq \min (\gamma(B), \gamma(C) ) \, .
			\end{align*}
		\end{proposition}
		\begin{proof}
			It is sufficient to prove that $\gamma(A) \leq \gamma(B)$ and $\gamma(A) \leq \gamma(C)$. We will show the first claim. The second one can be dealt with similarly.
			
			Let $\{(x_1,y_1),\dots,(x_t,y_t)\}$ be one of the maximum induced matchings of $A$. Using the same argument as in the proof of Proposition \ref{prop:subrankbound}, we can assume WLOG that $x_i = y_i = i, \forall i = 1, \ldots, t$. Therefore, the top left $t \times t$ submatrix of $A$ is a diagonal matrix.
			
			Since $A = BC$, then the $i$th column of $A$ is equal to $\sum_{k}C_{k,i}B_{:k}$, where $B_{:k}$ is the column $k$-th of $B$ and $C_{k,i}$ is the coefficients of $C$ at the index $(k,i)$. For each $1 \leq i \leq t$, we define $S^i = \{k \mid C_{k,i} > 0, B_{i,k} > 0\}$. It is clear that $S^i$ is non-empty (otherwise, the value $A_{i,i}$ is zero, a contradiction). Moreover, $\forall k \in S^i, \forall j \leq t, j \neq i$, $B_{j,k} =0$ to keep the submatrix at the top left corner of the matrix $A$ diagonal. As a consequence, $S^{i} \cap S^{j} = \emptyset$. Hence, we have $\{(1,y_1),\dots,(t,y_t)\}$ is an induced matching of $B$ where $y_i$ is an arbitrary element of $S^i$. Therefore $\gamma(A) \leq \gamma(B)$.  
		\end{proof}
		
		We can now proceed to the proof of Theorem~\ref{thm:nonnegative_subrank}.
		\begin{proof}[Proof of Theorem~\ref{thm:nonnegative_subrank}]
			
			From Proposition~\ref{prop:subrankbound} we have $\gamma(A) \leq \subnonrank(A)$ for any nonnegative matrix $A$. Suppose that $\subnonrank(A) = m$, we have $\id_{m} = XAY^{T}$, from Proposition~\ref{prop:ordermatching}, one has
			
			\begin{align*}
				m = \gamma(\id_{m}) =  \gamma(XAY^{T}) \leq \gamma(XA) \leq \gamma(A).
			\end{align*}
			That implies $\gamma(A) = \subnonrank(A) = m$.
		\end{proof}
		Theorem \ref{thm:nonnegative_subrank} characterizes a very nice property of nonnegative subrank: it is exactly equal to the maximum induced matching of the bipartite graph whose adjacency matrix is the support of its matrix. Finding the subrank of a matrix is thus reducible to the problem of finding the maximum induced matching of a bipartite graph. Unfortunately, the maximum induced matching problem is NP-hard, even in the setting of bipartite graph~\cite{cameron1989induced}. Even so, by using the properties of nonnegative matrices, we provide a combinatorial interpretation of the subrank of a nonnegative matrix. The subrank or symmetric subrank studied in other works \cite{Str91, christandl2021symmetric, christandl2022gap}, however, do not have this property. These quantities depend also on their coefficients.
		
		With this characterization of $\subnonrank$, we can now prove Theorem~\ref{thm:duality}.
		\begin{proof}[Proof of Theorem~\ref{thm:duality}]
			There are two cases to consider:
			\begin{enumerate}[leftmargin=*]
				\item If $A = 0$, then $\phi(A) = 0 $ for any $\phi \in \mathbf{X}(\M_{+},\leqk)$. The result follows trivially.
				\item If $A \neq 0$, it is difficult to directly apply Theorem~\ref{thm:Strassen_asymptotic} since it is non-trivial to verify the condition: There exists an integer $k$ such that $I_2 \leqk A^k$. In this proof, we provide a simple workaround: Instead of considering $A$, we consider $B:= (A \oplus A)$. It is trivial that $I_2 \leqk B^k$ with $k = 1$ because we assume that $A$ is non-zero. Therefore, using Theorem~\ref{thm:Strassen_asymptotic}, we obtain:
				\begin{align*}
					\asymrank(B)  = \max_{\phi \in \mathbf{X}(S,\leqslant)} \phi(B) \text{, and } \asymsubrank(B) = \min_{\phi \in \mathbf{X}(S,\leqslant)}\phi(B) \, .
				\end{align*}
				Due to the definition of $\mathbf{X}(S,\leqslant)$, we have: $\phi(B) = 2\phi(A), \forall \phi \in \mathbf{X}(S,\leqslant)$. Thus, we can rewrite the above equation as:
				\begin{equation}
					\label{eq:relation2times}
					\asymrank(B)  = 2\max_{\phi \in \mathbf{X}(S,\leqslant)} \phi(A) \text{, and } \asymsubrank(B) = 2\min_{\phi \in \mathbf{X}(S,\leqslant)}\phi(A) \, .
				\end{equation}
				It is thus, sufficient to show the two following equalities:
				\begin{align*}
					\nrank(B^k) & = 2^k \nrank(A^k)\\
					\subnonrank(B^k) & = 2^k \subnonrank(A^k).
				\end{align*}
				Using definitions of $\asynrank$ and $\asympsubnonrank$, we deduce that $\asynrank(B) = 2 \asynrank(A)$ and $\asympsubnonrank(B) = 2 \asympsubnonrank(A)$. Combining with Equation~\eqref{eq:relation2times}, we prove the theorem.
				
				To this end, we first observe that:
				\begin{equation}
					\label{eq:newtonexpansion}
					B^k = (A \oplus A)^k = \underbrace{A^k \oplus \ldots \oplus A^k}_{2^k \text{ times}}.
				\end{equation}
				For the case of $\nrank$, we will use the following claim: for every pair of nonnegative matrices $(X,Y)$, we have: $\nrank(X \oplus Y) = \nrank(X) + \nrank(Y)$. The proof of this claim can be found in \ref{appendix:techlem}. As a consequence of Equation~\eqref{eq:newtonexpansion}, we have: $\nrank(B^k) = 2^k \nrank(A^k)$.
				
				For the case of $\subnonrank$, thanks to Theorem~\ref{thm:nonnegative_subrank}, it is also clear that $\subnonrank$ is additive under $\oplus$, i.e., $\subnonrank(X \oplus Y) = \subnonrank(X) + \subnonrank(Y)$ for all nonnegative matrices $X, Y$. Indeed, the bipartite graph (corresponding to the support) of $(X \oplus Y)$ is the union of two disconnected bipartite graphs (corresponding to the support) of $X$ and $Y$ respectively. Thus, the maximum induced matching of $(X \oplus Y)$ is equal to the sum of these of $X$ and $Y$. Using the similar argument for $\nrank$, we also have $\subnonrank(B^k) = 2^k \subnonrank(A^k)$. This concludes the proof.\qedhere
			\end{enumerate}
			
			
		\end{proof}

		\section{Some elements in the asymptotic spectrum of nonnegative matrices}
		\label{sec4:elements_spectral_point}
		From Theorem~\ref{thm:duality}, an explicit description of the asymptotic spectrum of nonnegative matrices is desired to understand the asymptotic nonnegative rank. However, the abstract theory of asymptotic spectra in Section~\ref{sec2:spectra} does not give an explicit construction the elements of $\mathbf{X}(\M_{+},\leqk)$. Previous works for graphs and tensors, thus, are devoted to show some elements in their asymptotic spectrum. Namely, for the (quantum) asymptotic spectrum of graphs, several spectral points were introduced in~\cite{zuiddam2019asymptotic_graph,li2020quantum,gao2022tracial,gribling2020haemers} such as Lovász theta number~\cite{lovasz1979shannon}. In the case of tensors, \cite{Str91} provided a universal spectral point for the asymptotic spectrum of oblique tensors (oblique is a strict subset of all tensors). Recently, in~\cite{christandl2018universal} constructed an infinite family of the asymptotic spectrums of tensors over complex field. We recall that known all elements of the asymptotic spectrum of tensors can be used to determine the complexity of the matrix multiplication problem. In the same spirit, in this section, we show two  elements that belong asymptotic spectrum of nonnegative matrices $\mathbf{X}(\M_{+},\leqk)$ introduced in Section \ref{sec3:spectrum_nonnegative_matrices}, namely rank and fractional cover number. In addition, in the last subsection, we show that the asymptotic spectrum of lower triangular nonnegative matrices with non-zero values on their diagonal contains exactly a single point. 
		
		\subsection{Rank}
		We first remind readers of the definition of rank of a matrix.
		\begin{definition}
			Given a matrix $A \in \RR^{m \times n}$, the rank of the matrix $A$ (denoted as $\rank(A)$) is the largest integer $r$ such that there exists an invertible square submatrix of $A$ of size $r \times r$.
		\end{definition}
		The following lemma shows that $\rank$ is indeed a spectral point of $\mathbf{X}(\M_{+},\leqk)$.
		\begin{lemma}
			The matrix \emph{rank} is an element of the asymptotic spectrum $\mathbf{X}(\M_{+},\leqk)$. 
		\end{lemma}
		\begin{proof}
			Matrix rank is well-known to be multiplicative under $\otimes$ and additive under $\oplus$. Also clearly $\rank(\id_{n}) =n$ for all $n \in \NN$. Finally, it holds for any two matrices $X,Y$ that if $A = XBY$, using the fact that $\rank(AB) \leq \min \{\rank(A),\rank(B)\}$, then $\rank(A) \leq \rank(B)$. In particular, this is true if $X,Y$ are nonnegative matrices, which proves rank is monotone.
		\end{proof}	
		
		\subsection{Fractional cover number}
		The fractional cover number is a parameter of nonnegative matrices introduced by Lovász~\cite{LOVASZ1975383} to study the chromatic number of uniform hypergraphs. Furthermore, in~\cite{karchmer1995fractional}, the authors showed that this parameter gives a characterization for amortized
		nondeterministic communication complexity. In this subsection, we show that the fractional cover number is a spectral point of $\mathbf{X}(\M_{+},\leqk)$. 
		
		We first remind readers of the definition of fractional cover number.  Let $A \in \RR_{+}^{m\times n}$ be a nonnegative matrix, a subset $X\times Y \subseteq [m] \times [n]$ is called monochromatic rectangle of the matrix $A$ if $X \times Y \subseteq \supp(A)$. Let $R_1,\dots,R_{k}$ be all monochromatic rectangle of $A$. We define the \emph{fractional cover number} of $A$, denoted $F(A)$, as follows
		\begin{align*}
			\text{ minimize: } \quad &\sum_{i=1}^{k}\tau(R_i) \\
			\text{ s.t } &\sum_{i: (a,b) \in R_i}\tau(R_i) \geq 1 \text{ for all } (a,b) \in \supp(A) \, \\
			&0 \leq \tau(R_i) \leq 1 \,  \text{ for all } i \in [k] \, .
		\end{align*}
	The dual of the above linear program is given by: 
		\begin{align*}
			\text{ maximize: } \quad &\sum_{(a,b) \in \supp(A)} \mu(a,b) \\
			\text{ s.t } &\sum_{(a,b) \in R_i} \mu(a,b) \leq 1  \text{ for all } i \in [k] \, \\
			&0 \leq \mu(a,b) \leq 1 \, \text{ for all } (a,b) \in \supp(A) \, . 
		\end{align*}  	
		In~\cite{LOVASZ1975383}, the author showed that the fractional cover number has the following properties.
		\begin{enumerate}[label=(\roman*)]
			\item Multiplicative: $F(A\otimes B) = F(A)F(B)$.
			\item Additive: $F(A \oplus B) =  F(A)+F(B)$. 
		\end{enumerate}
		It is easy to verify that $F(\id_{n}) = n$ for all $n \in \NN$. To prove that $F$ is a spectral point, it is sufficient to prove its monotonicity with respect to the preorder $\leqk$, that is, if $A\leqk B$ then $F(A) \leq F(B)$. To do this, we first recall another characterization of fractional cover.
		
		We begin by reminding the notions of \emph{$t$-fold coverings} and \emph{$t$-fold covering number} of nonnegative $A$ \cite{scheinerman2011fractional}. Let $t$ be a positive integer, a $t$-fold covering of $A$ is a multiset $\{R_1,\dots,R_k\}$ (unlike set, a multi-set allows an element to appear multiple times) where each $R_i$ is a monochromatic rectangle of $A$ and each $(i,j) \in \supp(A)$ is covered by at least $t$ of the $R_i$'s. The smallest cardinality of such a multiset is called the $t$-fold covering number of $A$ and is denoted $F_{t}(A)$. One can describe $F_{t}(A)$ as an integer linear programming as:
		\begin{align*}
			\text{ minimize: } \quad &\sum_{i=1}^{k}\tau(R_i) \\
			\text{ s.t } &\sum_{i: (a,b) \in R_i}\tau(R_i) \geq t \text{ for all } (a,b) \in \supp(A) \, \\
			&\tau(R_i) \in \NN \,  \text{ for all } i \in [k] \, .
		\end{align*}
		In particular, $F(A)$ is simply a linear relaxation of the integer programming definition of $F_1(A)$. It is clear that $F_{t}$ is subadditive in its subscript, that is,
		\begin{align*}
			F_{s+t}(A) \leq F_{s}(A)+F_{t}(A),
		\end{align*} 
		because the sum of $s$ and $t$-fold coverings yields a $(s+t)$-fold coverings. By Fekete's lemma, we can define the \emph{asymptotic fractional covering number} of $A$ as follow
		\begin{align*}
			F^{*}(A) \coloneqq \lim_{t \ra \infty}\frac{F_t(A)}{t} = \inf_{t} \frac{F_{t}(A)}{t} \, . 
		\end{align*}
		
		\begin{proposition}[\cite{scheinerman2011fractional}]
			\label{prop:fractionalcharacterization}
			For any nonnegative matrix $A$, the fractional cover number is equal to the asymptotic fractional cover number, that is, $F^{*}(A) = F(A)$. 
		\end{proposition}
		\begin{proof}
			The proof can be found in~\cite{scheinerman2011fractional}, but we do it here for self-containedness.
			
			Let $R:= \{R_1,\dots,R_k\}$ be the set of all monochromatic rectangles of $A$. The fractional cover $F(A)$ can be described in the following compact form: 
			\begin{equation}
				\label{program:1}
				\begin{aligned}
					\text{minimize: } \sum_{i=1}^{k}x_{i} \\
					\text{s.t. } \quad Mx &\geq 1 \\
					x &\geq 0 \, . 		 
				\end{aligned}
			\end{equation}
			where $M$ is a binary matrix dependent on the set $R$. 
			
			First we prove $F(A) \leq F^{*}(A)$. Given an arbitrary positive integer $a$, let $\{R_1,\dots,R_v\}$ be the smallest $a$-fold covering of $A$ and let $x = (x_1, \ldots, x_k)$ the vector encoding the number of times $R_i$ appears in this covering. Then $Mx \geq a, x \geq 0$. This implies that $M(x/a) \geq 1$. Thus $x/a$ is a feasible solution of the linear program~\eqref{program:1}, hence $F(A) \leq \sum_{i=1}^{k}x_i/a = F_{a}(A)/a$. Since this holds for all positive integers $a$, we have $F(A) \leq F^{*}(A)$.
			
			Next we prove that $F(A) \geq F^{*}(A)$. Since all coefficients of the linear programming \eqref{program:1} are integers, there exists at least one rational vector $x$ that yields the minimum value of the linear program~\eqref{program:1}. Let $b$ be an integer such that $bx$ is an integer vector, we have $M(bx) \geq b, bx \geq 0$. This shows that we can form an $b$-fold cover of $A$ by choosing a monochromatic rectangle $R_i$ with multiplicity $bx_i$. Therefore $bF(A) = \sum_{i=1}^{k}bx_i \geq F_{b}(A)$. By definition of $F^*(A)$, we have $F(A) \geq F_b(A) / b \geq F^*(A)$.
		\end{proof}
		Using Proposition~\ref{prop:fractionalcharacterization}, we obtain the following result.
		\begin{proposition}
			\label{prop:fracorder}
			Let $A,B,C$ be nonnegative matrices, if $A = BC$ then $F(A) \leq \min(F(B),F(C))$.
		\end{proposition}
		\begin{proof}
			It is sufficient to prove that $F(A) \leq F(B), F(A) \leq F(C)$. We will show the former, i.e. $F(A) \leq F(B)$. The latter can be demonstrated similarly.
			
			In this proof, we represent a monochromatic rectangle by a rank-one binary matrix. In fact, given a subset $X \times Y \subset [m] \times [n]$, we represent it by the binary matrix $1_X \times 1_Y^\top \in \RR_{+}^{m \times n}$ where $1_X \in \RR_{+}^m$ (resp. $1_Y \in \RR_{+}^n$) is the indication vector of the set $X$ (resp. $Y$).
			
			If $A = BC$, we will show that $F_{t}(A) \leq F_t(B)$ for any positive integer $t$. That concludes the proof by using the characterization of Proposition \ref{prop:fractionalcharacterization}.
			
			Let $\{R_1,\dots,R_k\}$ be a multiset of monochromatic rectangles that yields the minimum value for $t$-fold covering of $B$. As discussed, there exist two binary vectors $x_i, y_i$ such that $R_i = x_i y_i^\top$. 
			
			We define $R'_i = \supp(R_iC)$ ($R_iC$ is the product between two nonnegative matrices $R_i$ and $C$). We prove that the set $\{R_1', \ldots,  R_k'\}$ is a $t$-fold coverings of $A$. It is immediate that $F_t(A) \leq k = F_t(B)$. To do so, we first demonstrate that $R_i'$ is a monochromatic rectangle of $A$. Indeed, we use the following property of nonnegative matrices: for $M,N$ two nonnegative matrices, then:
			\begin{equation}
				\label{eq:supportrelation}
				\supp(MN) = \supp(\supp(M)\supp(N)),
			\end{equation}
			where $\supp(M)$ is the binary matrix representation of the support of the matrix $M$. In other words, the support of the product of $MN$ is determined by the product of supports of $M$ and $N$ (in their binary matrix forms). This claim only holds for nonnegative matrices and it is incorrect for general matrices.
			
			First, it is clear that $R'_i \subseteq \supp(A)$. Indeed,
			\begin{align*}
				R'_i= \supp(R_iC) &= \supp(\supp(R_i)\supp(C))\\
				&\subseteq \supp(\supp(B)\supp(C)) \qquad \text{(since $R_i \subseteq \supp(B)$)}\\
				&\overset{\eqref{eq:supportrelation}}{=} \supp(BC) = \supp(A).
			\end{align*}
			Moreover, 
			$$R'_i = \supp(R_iC) = \supp(x_iy_i^\top C) = \supp(\supp(x_i)\supp(y_i^\top C)) = \supp(x_i)\supp(y_i^\top C)$$
			since $\supp(x_i)\supp(y_i^\top C)$ is a binary matrix (product of two binary vectors). Thus, $R_i'$ is indeed a monochromatic rectangle of $A$.

			For any entry $(i,j) \in \supp(A)$, since $0<A_{ij} = \sum_{z=1}B_{i,z}C_{z,j}$, there is an index $z$ such that  $B_{i,z}>0$ and $C_{z,j}>0$. Suppose that $(i,z) \in \supp(B)$ is covered by a multiset $\{R_{i_1},\dots,R_{i_t} \}$ then the entry $(i,j) \in \supp(A)$ is covered by the multiset $\{R'_{i_1},\dots,R'_{i_t} \}$. Therefore, $\{R'_1, \dots, R'_k\}$ are indeed a $t$-fold coverings of $A$. That concludes the proof. \qedhere
		\end{proof}
		
		\begin{proposition}
		    If $A\leqk B$ then $F(A) \leq F(B)$.
		\end{proposition}
		\begin{proof}
			From $A \leqk B$, we have $A = XBY^{T}$. By the Proposition~\ref{prop:fracorder}, one has
			\begin{align*}
				F(A) = F(XBY^{T}) \leq F(XB)  \leq F(B) \, .
			\end{align*}     
		\end{proof}
		We have the following lemma.
		\begin{lemma}
		    The fractional cover number $F$ is an element of the asymptotic spectrum $\mathbf{X}(\M_{+},\leqk)$. 
		\end{lemma}
		\begin{remark} \textbf{Rank and fractional cover number are incomparable}. We have two elements in asymptotic spectrum $\mathbf{X}(\M_{+},\leqk)$ are \emph{rank} and \emph{fractional cover number} $F$. The following examples show that two quantities are incomparable, i.e., there are nonnegative matrices $A,B$ with $F(A) > \rank(A)$  and $F(B) < \rank(B)$.
			\begin{equation*}
				A = 
				\begin{pmatrix}
					1 & 1 & 0 & 0 \\
					1 & 0 & 1 & 0 \\
					0 & 1 & 0 & 1 \\
					0 & 0 & 1 & 1
				\end{pmatrix}
			\end{equation*}
			We have $F(A) = 4$, but $\rank(A) = 3$.

			\begin{equation*}
				B = 
				\begin{pmatrix}
					1 & 0 & 1 & 1 \\
					0 & 1 & 1 & 0 \\
					0 & 1 & 1 & 1 \\
					1 & 1 & 0 & 1
				\end{pmatrix}
			\end{equation*}
			On can compute that $F(B) = 3.5$, and $\rank(B) = 4$.
		\end{remark}

		\subsection{Asymptotic spectrum of triangular nonnegative matrices}
		
		As defined in Section~\ref{sec3:spectrum_nonnegative_matrices}, asymptotic spectrum is a $\leqk$ monotone semiring homomorphism from $S$ to $\RR_{\geq 0}$ that satisfies the constraint $\phi(\id_n) = n, \forall n \in \NN, \forall \phi \in \mathbf{X}(\M_{+},\leqk)$. In other words, the set of $\{\phi(\id_n) \mid \phi \in \mathbf{X}(\M_{+},\leqk)\}$ collapses to a single value $n$. In this section, we prove that this phenomenon also happens to the set of lower (square) triangular matrices with nonzero diagonal. 
		
		Let $A \in \RR_{+}^{m\times n}$ be a nonnegative matrix, we call the matrix $A$ is a \emph{nonnegative rectangle lower triangular} matrix if $A_{i,j} = 0$ for all $i>j$.  
		\begin{lemma}
			\label{lem:Triagle-rank}
			Let $A \in \RR_{+}^{m\times n}$ be a rectangle lower triagular nonnegative matrix. Define $V = \{ i \leq \min \{m,n\}: A_{i,i} \neq 0 \}$, then $\asympsubnonrank(A) \geq |V|$. 
		\end{lemma}
		\begin{proof}
			Let $d = |V|$ and let $N$ be any multiple of $d$. Our proof is constructive, i.e, for any $N$ multiple of $d$, we can find a \emph{diagonal} submatrix of $A^{\otimes N}$ whose coefficients on the diagonal are nonzero of size $d_N$ such that $\lim_{N \to \infty} d_N^{1/N} = d$. Our argument is valid since Proposition \ref{thm:nonnegative_subrank} assure that $\subnonrank(A^{\otimes N}) \geq d_N$. 
			
			We argue that one can assume that the matrix $A$ is square $A$, i.e. $m = n$. Indeed, if $A$ is not square, one can take the principle submatrix $B$ of size $p \times p$ of the matrix $A$ where $p = \min(m,n)$ and prove that $\asympsubnonrank(B) \geq |V|$. Since $\asympsubnonrank(A) \geq \asympsubnonrank(B)$ (by Proposition \ref{thm:nonnegative_subrank}), we have: $\asympsubnonrank(A) \geq |V|$.
			
			Hence, in the remaining proof, we assume $m = n$ and the size of the matrix $A^{\otimes N}$ is thus $n^N \times n^N$. In this proof, we will use $\cI := \{0, \ldots, n^N - 1\}$ as the set of row (or column) index of $A^{\otimes N}$. It is convenient to represent each element of $\cI$ in $n$-base (or radix), i.e, given $i \in \cI_R$, we write $i = (i_1, \ldots, i_N), 0 \leq i \leq n - 1$ to implies: $i = \sum_{k = 1}^{N} i_k n^{k - 1}$. This representation is useful since the coefficient at the index $(i,j)$ of $A^{\otimes N}$ has the form:
			\begin{equation}
				\label{eq:productrep}
				A^{\otimes N}_{i,j} = \prod_{k = 1}^N A_{i_k, j_k}.
			\end{equation}
			
			Define $S(N)$ the subset of $\supp(A^{\otimes N})$ \emph{on the diagonal} such that for all $(i,i) \in S(N), i = (i_1,\dots,i_N) \in S(N)$, we have:
			\begin{align*}
				S(N) \coloneqq \{(i,i) \in \supp(A^{\otimes N}) \mid |\{k \in [N]: i_k = p\}| = N/d \text{ for all } p \in V \}
			\end{align*} 
			In words, if $(i,i) \in S(N)$, then the $n$-basis representation of $i$ contains exactly $N/d$ digits $p$ for every $p \in V$. The size of $S(N)$ is thus,
			\begin{equation*}
				|S(N)| = \binom{N}{N/d,\dots,N/d} \geq \frac{d^N}{(N+1)^{d}} \, .
			\end{equation*}
			The inequality follows from the fact that the largest multinomial coefficient is the central one, that is, $\binom{N}{a_1,\dots,a_d} \leq \binom{N}{\frac{N}{d},\dots, \frac{N}{d}}$ and the number of possible partitions of $N$ into $d$ parts is at most $(N+1)^{d}$. This proves $\asympsubnonrank(A) \geq \lim_{N \to \infty} |S(N)|^{1/N} = d$. Thus, it is sufficient to show that the \emph{principle} submatrix of $A$ restricted to indices $S(N)$ is a  diagonal matrix with nonzero coefficients in its diagonal.
			
			Consider two indices $i,j \in S(N)$  with $i = (i_1, \ldots, i_N), j = (j_1, \ldots, j_N)$ and $i \neq j$. Since $\sum_{k = 1}^N i_k = \sum_{k = 1}^N y_k = (N/d) \sum_{p \in V} p$, there must exist an index $\ell \leq N$ such that $i_\ell > j_\ell$. By Equation \ref{eq:productrep}, we deduce that $A_{i,j} = 0$ (since $A_{i_\ell, j_\ell} = 0$). In addition, $A_{i,i} = \prod_{k = 1}^N A_{i_k, i_k} \neq 0, \forall i \in S(N)$ because $i_k \in V$. That shows elements on the diagonal are nonzero. That concludes the proof.
		\end{proof}

		Denote $T \subseteq \M_{+}$ is a subsemiring which generated by all lower triangular matrices with nonzero entries on diagonal under operation $\oplus$ and $\otimes$. From Lemma~\ref{lem:Triagle-rank}, the asymptotic spectrum $\mathbf{X}(T,\leqk)$ contains only a single point. It is a corollary of Lemma \ref{lem:Triagle-rank}.
		
		\begin{corollary}
			Let $A \in \RR_{+}^{m\times n}$ be a lower rectangular triangular nonnegative matrix with nonzero on its diagonal. Then $\phi(A) = \min(m,n)$ for any $\phi \in \mathbf{X}(\M_{+},\leqk)$.  
		\end{corollary}
		\begin{proof}
			Let $d = \min(m,n)$, since $A\in \RR_{+}^{m \times n}$ with nonzero on its diagonal, it implies $\rank(A) = \nrank(A) = d$, therefore $\asynrank(A) \leq d$. By Lemma~\ref{lem:Triagle-rank} one has $\asympsubnonrank(A) \geq d$. Using Theorem \ref{thm:Strassen_asymptotic} (which states $\asynrank(A) = \max_{\phi \in \mathbf{X}(\M_{+},\leqk)} \phi(A) \geq \min_{\phi \in \mathbf{X}(\M_{+},\leqk)} \phi(A) = \asympsubnonrank(A)$), we have $\asynrank(A) = \asympsubnonrank(A) = d$ and $\phi(A) = d, \forall \phi \in \mathbf{X}(\M_{+},\leqk)$. This proves the corollary. 
		\end{proof}

		\paragraph{Acknowledgments}  HT would like to thank Jeroen Zuiddam for the useful discussions about the theory of asymptotic spectra and asymptotic nonnegative rank. We also thank Yaroslav Shitov and Vincent Y. F. Tan for the useful discussions about the non-multiplicativity of nonnegative rank and the exact Rényi common information.

\newpage
\phantomsection
\addcontentsline{toc}{section}{Bibliography}
\bibliographystyle{alphaurl}
\bibliography{bibliofile}
\newpage 
 \appendix
 \section{Proof of technical lemmas}
 \label{appendix:techlem}
 \begin{theorem}
 	For every pair of nonnegative matrices $(X,Y)$, we have:
 	\begin{equation*}
 		\nrank(X \oplus Y) = \nrank(X) + \nrank(Y)
 	\end{equation*}
 \end{theorem}
 \begin{proof}
 	Let $m = \nrank(X)$ and $n = \nrank(Y)$. By definition, there exist sets of pairs of nonnegative vectors: $\{(a_i,b_i) \mid i = 1, \ldots, m\}$ and $\{(c_i, d_i) \mid i = 1, \ldots, n\}$ such that:
 	\begin{equation}
 		\begin{aligned}
 			X & = \sum_{i = 1}^m a_ib_i^\top, \qquad
 			Y & = \sum_{i = 1}^n c_id_i^\top
 		\end{aligned}
 		\label{eq:nrankdef}
 	\end{equation}
 	
 	To prove the equality, it is sufficient to show that $\nrank(X \oplus Y) \leq m + n$ and $\nrank(X \oplus Y) \geq m + n$. Both inequalities are shown below:
 	\begin{enumerate}
 		\item Proof of $\nrank(X \oplus Y) \leq m + n$: This claim is trivial because:
 		\begin{equation*}
 			X \oplus Y = \begin{pmatrix}
 				X & 0\\
 				0 & Y
 			\end{pmatrix} = \sum_{i = 1}^m \bar{a}_i\bar{b}_i^\top + \sum_{i = 1}^n \bar{c}_i\bar{d}_i^\top
 		\end{equation*}
 		where $\bar{a}_i = \left[\begin{smallmatrix}
 			a_i \\ 0
 		\end{smallmatrix}\right]$ and $b_i = \left[\begin{smallmatrix}
 			b_i \\ 0
 		\end{smallmatrix}\right], \forall i = 1, \ldots, m$, the concatenation of vectors $a_i, b_i$ with a zero vector respectively. Similarly, $\bar{c}_i = \left[\begin{smallmatrix}
 			0 \\ c_i
 		\end{smallmatrix}\right]$ and $\bar{d}_i = \left[\begin{smallmatrix}
 			0 \\ d_i
 		\end{smallmatrix}\right], \forall i = 1, \ldots, n$. Since $\bar{a}_i, \bar{b}_i, \bar{c}_i, \bar{d}_i$ are nonnegative by construction, we have: $\nrank(X \oplus Y) \leq m + n$.
 		
 		\item Proof of $\nrank(X \oplus Y) \geq m + n$: Assume that $\nrank(X \oplus Y) = \ell$. By definition of nonnegative rank, there exists a set of pairs of nonnegative vectors $(x_i, y_i)$ such that:
 		\begin{equation}
 			\label{eq:nrankdef2}
 			X \oplus Y = \sum_{i = 1}^\ell x_iy_i^\top
 		\end{equation}
 		Note that the ``top right'' and the ``bottom left'' of $X \oplus Y$ are zero matrices, we must have either: $\supp(x_iy_i^\top) \subseteq \supp(\left[\begin{smallmatrix}
 			X & 0\\
 			0 & 0
 		\end{smallmatrix}\right])$ or $\supp(x_iy_i^\top) \subseteq \supp(\left[\begin{smallmatrix}
 			0 & 0\\
 			0 & Y
 		\end{smallmatrix}\right])$. As a consequence, each pair $(x_i, y_i)$ can be written either:
 		\begin{align}
 			(x_i, y_i) &= \left(\begin{bmatrix}
 				\bar{x}_i \\ 0
 			\end{bmatrix}, \begin{bmatrix}
 				\bar{y}_i \\ 0
 			\end{bmatrix}\right) \label{eq:form1}\\ \text{ or } \quad (x_i, y_i) &= \left(\begin{bmatrix}
 				0 \\ \bar{x}_i
 			\end{bmatrix}, \begin{bmatrix}
 				0 \\ \bar{y}_i
 			\end{bmatrix}\right) \label{eq:form2}
 		\end{align}
 		where $(\bar{x}_i, \bar{y}_i)$ are nonnegative vectors.
 		
 		One can thus partition the set $\{1, \ldots, \ell\}$ into two sets $\cX$ and $\cY$ such that $i \in \cX$ (resp. $i \in \cY$) if $(x_i,y_i)$ has the form in Equation~\eqref{eq:form1} (resp. Equation~\eqref{eq:form2}). To make Equation~\eqref{eq:nrankdef2}, it is necessary that:
 		\begin{equation*}
 			X = \sum_{i \in \cX} \bar{x}_i\bar{y}_i^\top, \qquad Y = \sum_{i \in \cY} \bar{x}_i\bar{y}_i^\top
 		\end{equation*}
 		By definition of nonnegative rank, we have: $|\cX| \geq m, |\cY| \geq n$. Thus, $\ell = |\cX| + |\cY| = m + n$. This concludes the proof. \qedhere
 	\end{enumerate}
 \end{proof}
\end{document}